\documentclass[11pt]{article}

\usepackage{a4}
\usepackage[top=30truemm,bottom=30truemm]{geometry}
\usepackage{authblk}
\usepackage{amsthm}
\usepackage{amsmath,amssymb}
\usepackage{tabls}
\usepackage[dvipdfmx]{graphicx}
\usepackage{array}
\usepackage{url}
\usepackage{cases}
\usepackage[normalem]{ulem}

\usepackage[draft,mode=multiuser]{fixme}
\fxuseenvlayout{colorsig}
\fxusetargetlayout{color}
\FXRegisterAuthor{an}{aan}{AN}
\FXRegisterAuthor{yn}{ayn}{YN}
\FXRegisterAuthor{si}{asi}{SI}
\FXRegisterAuthor{hb}{ahb}{HB}
\FXRegisterAuthor{mt}{amt}{MT}
\fxsetface{env}{}

\newtheorem{theorem}{Theorem}
\newtheorem{lemma}{Lemma}

\newcommand{\Factors}{\mathit{Factors}}

\newcommand{\Occ}{\mathsf{Occ}}

\newcommand{\STree}{\mathsf{STree}}

\newcommand{\slink}{\mathsf{slink}}

\newcommand{\SA}{\mathsf{SA}}
\newcommand{\LCP}{\mathsf{LCP}}

\newcommand{\LG}{\mathsf{LGOcc}}
\newcommand{\len}{\mathit{len}}
\newcommand{\occ}

\newcommand{\bit}{\mathsf{bit}}

\title{
  $O(n \log n)$-time text compression \\ by LZ-style longest first substitution
}

\author{Akihiro~Nishi}
\author{Yuto~Nakashima}
\author{Shunsuke~Inenaga}
\author{Hideo~Bannai}
\author{Masayuki~Takeda}

\affil{\textit{Department of Informatics, Kyushu University, Japan} \\
\texttt{\small \{akihiro.nishi, yuto.nakashima, inenaga, bannai, takeda\}@inf.kyushu-u.ac.jp}
}

\date{}

\begin{document}
\maketitle

\begin{abstract}
  Mauer et al. [A Lempel-Ziv-style Compression Method for Repetitive Texts,
  PSC 2017] proposed a hybrid text compression method called \emph{LZ-LFS}
  which has both features of Lempel-Ziv 77 factorization and longest first
  substitution.
  They showed that LZ-LFS can achieve better compression ratio
  for repetitive texts, compared to some state-of-the-art compression
  algorithms.
  The drawback of Mauer et al.'s method is that
  their LZ-LFS compression algorithm takes $O(n^2)$ time
  on an input string of length $n$.
  In this paper, we show a faster LZ-LFS compression algorithm
  that works in $O(n \log n)$ time.
  We also propose a simpler version of LZ-LFS that can be computed
  in $O(n)$ time.
\end{abstract}

\section{Introduction}

\emph{Text compression} is a task to
compute a small representation of an input text (or string).
Given a vast amount of textual data that has been produced to date,
text compression can play central roles in
saving memory space and reducing data transmission costs.

\emph{Lempel-Ziv 77} (\emph{LZ77})~\cite{LZ77}
is a fundamental text compression method
that is based on a greedy factorization of the input string.
LZ77 factorizes a given string $w$ of length $n$
into a sequence of non-empty substrings $f_1, \ldots, f_k$ such that
(1) $w = f_1 \cdots f_k$ and
(2) each factor $f_i$ is the longest prefix of $w[|f_1 \cdots f_{i-1}|+1..n]$
that has an occurrence beginning at a position
in range $[1..|f_1 \cdots f_{i-1}|]$ (this is a self-reference variant),
or $f_i = c$ if it is the leftmost occurrence of the character $c$ in $w$.
Each factor $f_i$ in the first case is encoded
as a reference pointer to one of its previous occurrences in the string.
LZ77 and its variants are basis of many text compression programmes,
such as gzip.

In the last two decades, \emph{grammar compression} has also gathered much attention.
Grammar compression finds a small context-free grammar
which generates only the input string.
Since finding the smallest grammar representing a given string is NP-hard~\cite{LZSS,Storer77},
various kinds of efficiently-computable greedy grammar compression algorithms
have been proposed.
The most well-known method called Re-pair~\cite{LarssonDCC99}
is based on a most frequent first substitution approach,
such that most frequently occurring bigrams (substrings of length 2)
are replaced with new non-terminal symbols recursively,
until there are no bigrams with at least two non-overlapping occurrences.
An alternative is a \emph{longest first substitution} (\emph{LFS}) approach,
where longest substrings that have at least two non-overlapping
occurrences are replaced with new non-terminal symbols recursively,
until there are no substrings of length at least two with
at least two non-overlapping occurrences.

Recently,
Mauer et al.~\cite{MauerBO17} proposed a hybrid
text compression algorithm called \emph{LZ-LFS},
which has both features of LZ77 and LFS.
Namely, LZ-LFS finds a longest substring which occurs
at least twice in the string,
replaces its selected occurrences with a special symbol $\#$,
and encodes each of them as a reference to its leftmost occurrence.
This is continued recursively, until there are no substrings of length
at least two which occur at least twice in the string.
The details on how the occurrences to replace are selected
can be found in~\cite{MauerBO17}
as well as in a subsequent section in this paper.
Mauer et al. showed that LZ-LFS can have good practical performance
in compressing repetitive texts.
Indeed, in their experiments, 
the compression ratio of LZ-LFS outperforms
that of some state-of-the-art compression algorithms
on data sets from widely-used corpora.
The drawback, however, is that Mauer et al.'s compression algorithm
for LZ-LFS takes $O(n^2)$ time for input strings of length $n$.

In this paper, we focus on a theoretical complexity for
computing LZ-LFS, and propose a faster LZ-LFS algorithm
which runs in $O(n \log n)$ time with $O(n)$ space.
Our algorithm is based on Nakamura et al.'s algorithm
for LFS-based grammar compression~\cite{nakamura09:_linear_longest_first_}.
Although Nakamura et al.'s algorithm is quite involved,
our algorithm for LZ-LFS is much less involved
due to useful properties of LZ-LFS.
We also show that a simplified version of LZ-LFS
can be computed in $O(n)$ time and space with slight modifications
to our algorithm.

\section{Preliminaries}

\subsection{String notations}

Let $\Sigma$ be an alphabet.
An element of $\Sigma^*$ is called a \emph{string}.
Strings $x$, $y$, and $z$ are said to be a \emph{prefix}, 
\emph{substring}, and \emph{suffix} of string $w=xyz$, respectively.

The length of a string $w$ is denoted by $|w|$. 
The empty string is denoted by $\varepsilon$, that is, $|\varepsilon|=0$.
Let $\Sigma^+=\Sigma^* \setminus \{\varepsilon\}$.
The $i$-th character of a string $w$ is denoted by $w[i]$  for $1 \leq i \leq |w|$,
and the substring of a string $w$ that begins at position $i$ and
ends at position $j$ is denoted by $w[i..j]$ for $1 \leq i \leq j \leq |w|$.
For convenience, let $w[i..j]=\varepsilon$ for $j < i$, and $w[i..] = w[i..|w|]$ for $1\leq i \leq |w|$.

An \emph{occurrence} of a substring $x$ of a string $w$
is an interval $[i..i+|x|-1]$ such that $w[i..i+|x|-1] = x$.
For simplicity, we will sometimes call the beginning position $i$
of $x$ as an occurrence of $x$ in $w$.
Let $\Occ_w(x)$ denote the set of the beginning positions
of the occurrences of $x$ in $w$.
If $x$ does not occur in $w$, then $\Occ_w(x) = \emptyset$. 

If $|\Occ_w(x)| \geq 2$, then $x$ is said to be a \emph{repeat} of $w$.
A repeat $x$ of $w$ is said to be a \emph{longest repeat} (\emph{LR}) of $w$
if there are no repeats of $w$ that are longer than $x$.
We remark that there can exist more than one LR for $w$ in general.
A repeat $y$ of $w$ is said to be a \emph{maximal repeat} of $w$
if for any characters $a,b \in \Sigma$,
$|\Occ_{w}(ay)| < |\Occ_{w}(y)|$ and $|\Occ_{w}(yb)| < |\Occ_w(y)|$.
We also remark that any longest repeat of $w$
is a maximal repeat of $w$.

Let $I = \{i_1, \ldots, i_k\} \subseteq \Occ_w(x)$
be a (sub)set of occurrences of a repeat $x$ in $w$
such that $k \geq 2$ and $i_1 < \cdots < i_k$.
The occurrences in $I$ are said to be \emph{overlapping}
if $i_1 + |x| - 1 \geq i_k$,
and are said to be \emph{non-overlapping}
if $i_j + |x| - 1 < i_{j+1}$ for all $1 \leq j < k$.

\subsection{Suffix trees}

Assume that any string $w$ terminates with a unique symbol $\$$
which does not occur elsewhere in $w$.
The \emph{suffix tree} of a string $w$, denoted $\STree(w)$,
is a path-compressed trie
such that each edge is labeled with a non-empty substring of a string of $w$,
each internal node has at least two children,
the labels of all out-going edges of each node begin with
mutually distinct characters,
and each suffix of $w$ is spelled out by a path
starting from the root and ending at a leaf.
Because we have assumed that $w$ terminates with a unique symbol $\$$,
there is a one-to-one correspondence between
the suffixes of $w$ and the leaves of $\STree(w)$.
The \emph{id} of a leaf of $\STree(w)$ is defined to be
the beginning position of the suffix of $w$ that it represents.

Each node of $\STree(w)$ is specifically called
as an \emph{explicit} node,
and in contrast a locus on an edge is called as an \emph{implicit} node.
For ease of explanation,
we will sometimes identify each node of $\STree(w)$
with the string obtained by concatenating the edge labels
from the root to that node.
In the sequel, the \emph{string depth} of a node implies
the length of the string that the node represents.

Each edge label $x$ is represented by a pair $(i, j)$ of positions
in $w$ such that $w[i..j] = x$,
and in this way $\STree(w)$ can be represented with $O(n)$ space.
Every explicit node $v$ of $\STree(w)$ except for the root node has
an auxiliary reversed edge called the {\em suffix link},
denoted $\slink(v)$, such that $\slink(v) = v'$
iff $v'$ is a suffix of $v$ and $|v'|+1 = |v|$.
Notice that if $v$ is a node of $\STree(w)$,
then such node $v'$ always exists in $\STree(w)$.
$\STree(w)$ can be constructed in $O(n)$ time and space
if a given string $w$ of length $n$ is drawn
from an integer alphabet of size $n^{O(1)}$~\cite{Farach-ColtonFM00},
or in $O(n \log \sigma)$ time and $O(n)$ space
if $w$ is drawn from a general ordered alphabet
and $w$ contains $\sigma$ distinct characters~\cite{Weiner,McC76,Ukk95}.

\section{Text compression by LZ-style longest first substitution}
\label{sec:LZ-LFS_prelim}

Mauer et al.~\cite{MauerBO17} proposed
a text compression method which is a hybrid of
the Lempel-Ziv 77 encoding (LZ)~\cite{LZ77} and
a grammar compression with
longest first substitution (LFS)~\cite{nakamura09:_linear_longest_first_},
which hereby is called \emph{LZ-LFS}.

\subsection{LZ-LFS}

Here we describe how LZ-LFS compresses a given string $w$.

Let $x$ be an LR of $w$,
and let $\ell$ be the leftmost occurrence of $x$ in $w$.
Let $\LG_w(x)$ denote the set of
non-overlapping occurrences of $x$ in $w$
that are selected in a left-greedy manner (i.e., greedily from left to right).
Notice that $\ell = \min(\LG_w(x)) = \min(\Occ_w(x))$.
An occurrence $i$ of $w$ is said to be of
\begin{itemize}
\item Type 1 if $i$ is the second leftmost occurrence of $x$
  (i.e., $i = \min(\Occ_w(x) \setminus \{\ell\})$) and
  the occurrences $\ell$ and $i$ overlap
  (i.e., $\ell+|x| -1 \geq i$).
\end{itemize}
Let $\ell'$ be the Type 1 occurrence of $x$ in $w$ if it exists,
and let
\begin{equation}
e =
\begin{cases}
  \ell'+|x|-1 & \mbox{if $\ell'$ exists,}\\
  \ell+|x|-1 & \mbox{otherwise.}
\end{cases}
\label{eqn:e}
\end{equation}
An occurrence $i$ of $x$ in $w$ is said to be of
\begin{itemize}
\item Type 2 if $i$ is the leftmost occurrence of $x$
  after $e$ and
  there is no non-overlapping occurrence of $x$ to the right of $i$
  (i.e., $\{i\} = \LG_{w[e+1..]}(x)$).

\item Type 3 if $i$ is a left-greedily selected
  occurrence of $x$ after $e$ (i.e., $i \in \LG_{w[e+1..]}(x)$) and
  there are at least two such occurrences of $x$
  (i.e., $|\LG_{w[e+1..]}(x)| \geq 2$).

\item Type 4 otherwise.
\end{itemize}
Note that Type 2 and Type 3 occurrences of $x$ cannot simultaneously exist.

LZ-LFS is a recursive greedy text compression method which works as follows:
Given an input string $w$,
LZ-LFS first finds an LR $x$ of $w$ and picks up
its Type 1 occurrence (if it exists),
and either its Type 2 occurrence or its Type 3 occurrences.
Each of these selected occurrences of $x$ is
replaced with a special symbol $\#$ not appearing in $w$,
together with a pointer to the leftmost occurrence $\ell$ of $x$
which still remains in the modified string.
The encoding of this pointer differs for each type of occurrences,
see~\cite{MauerBO17} for details.
We remark that Type 4 occurrences are not selected for replacement
and all the Type 4 occurrences but
the leftmost occurrence of $x$ disappear in the modified string.
In the next step, LZ-LFS finds an LR of the modified string
which does \emph{not} include $\#$,
and performs the same procedure as long as
there is a repeat in the modified string.

Let $w_k$ denote the modified string in the $k$th step.
Namely, $w_0 = w$ and $w_{k}$ is the string after
all the selected occurrences of an LR of $w_{k-1}$
have been replaced with $\#$.
LZ-LFS terminates when it encounters the smallest $m$
such that $w_m$ does not contain
repeats of length at least two which consists only of
characters from the original string $w$ (i.e., repeats without $\#$'s).

LZ-LFS computes a list $\Factors$ as follows:
Initially, $\Factors$ is an empty list.
For each occurrence $i$ of LR $x$ that has been replaced with $\#$,
a pair $(\ell, |x|)$ of its leftmost occurrence $\ell$
and the length $|x|$ is added to $\Factors$
if it is of Type 2 or the first occurrence of Type 3.
Otherwise (if it is of Type 1),
then a pair $(i-\ell, |x|)$ is added to $\Factors$.
These pairs are arranged in $\Factors$ in increasing order
of the corresponding occurrences in the input string.

LZ-LFS also computes an array $F$ as follows:
Suppose we have computed $w' = w_m$.
For each $1 \leq h \leq |F|$,
if the $h$-th $\#$ from the left in $w'$ replaced
a Type 1 occurrence of an LR, then $F[h] = 1$.
Similarly, if the $h$-th $\#$ from the left in $w'$ replaced
a Type 2 occurrence of an LR, then $F[h] = 2$.
For Type 3 occurrences,
$F[h] = 2+j$ if the $h$-th $\#$ from the left in $w'$
replaced the $j$-th LR that that has Type 3 occurrences.
This array $F$ can be computed e.g.,
by using an auxiliary array $A$ of length $n$,
where each entry is initialized to null.
For each occurrence $i$ of each LR $x$ that has been replaced with $\#$,
the type of the occurrence (Type 1, 2, or 3) is stored at $A[i]$.
After the final string $w' = w_m$ has been found,
non-null values of $A$ are extracted by a left-to-right scan,
and are stored in $F$ from left to right.
A tuple $(w', \Factors, F)$ is the output of
the compression phase of LZ-LFS.

To see how LZ-LFS compresses a given string,
let us consider a concrete example with string
\[
w = w_1 = \mathtt{abcabcaabcdabcacabc\$}.
\]
There are two LRs $\mathtt{abca}$ and $\mathtt{cabc}$ in $w$,
and suppose that $\mathtt{abca}$ has been selected to replace.
Below, we highlight the occurrences of $\mathtt{abca}$
with underlines:
\begin{eqnarray*}
  w_1 & = & \mathtt{abc\uwave{abca}abcd\uline{abca}cabc\$}. \\
   & & \setlength\ULdepth{-4mm}\uuline{\hspace{9mm}}
\end{eqnarray*}
\vspace*{-1cm}

The wavy-underlined occurrence of $\mathtt{abca}$ at position $4$
is of Type 1 since it overlaps with the leftmost occurrence of $\mathtt{abca}$
which is doubly underlined.
Then, pair $(3, 4)$ is added to $\Factors$,
where the first term $3$ is the distance from the occurrence
at position $4$ to the leftmost occurrence at position $1$,
and the second term $4$ is $|\mathtt{abca}|$.

The singly underlined occurrence of $\mathtt{abca}$
at position $12$ is of Type 2 since
it does not overlap with the leftmost occurrence of $\mathtt{abca}$,
and there are no occurrences of $\mathtt{abca}$ to its right.
Then, pair $(1, 4)$ is added to $\Factors$,
where $1$ is the leftmost occurrence of $\mathtt{abca}$ and $4 = |\mathtt{abca}|$.

These Type 1 and Type 2 occurrences of $\mathtt{abca}$ are replaced with
with $\#$, and the resulting string is
\[
w_2 = \mathtt{\uuline{abc}\#\uline{abc}d\#c\uline{abc}\$},
\]
of which $\mathtt{abc}$ is an LR.
Since neither the second occurrence nor the third one of $\mathtt{abc}$
overlaps with the leftmost occurrence of $\mathtt{abc}$,
both of these occurrences are of Type 3.
Hence, pair $(1, 3)$ is added to $\Factors$,
where $1$ is the leftmost occurrence of $\mathtt{abc}$
and $3 = |\mathtt{abc}|$.
Finally, we obtain
\[
w_3 = \mathtt{abc\#\#d\#c\#\$}.
\]
Since $w_3$ has no repeats of length at least two which does not contain $\#$'s,
LZ-LFS terminates here.
Together with this final string $w' = w_3$,
LZ-LFS outputs $\Factors = \langle (3,4),(1,3),(1,4) \rangle$
and $F = [1,3,2,3]$.
Recall that the pairs in $\Factors$ are arranged in
increasing order of the corresponding occurrences in the input string $w$.

Mauer et al.~\cite{MauerBO17} showed how to decompress
$(w', \Factors, F)$ to get the original string $w$ in $O(n)$ time.
On the other hand, Mauer et al.'s LZ-LFS compression algorithm for
computing $(w', \Factors, F)$ from the input string
$w$ of length $n$ uses $O(n^2)$ time and $O(n)$ space.
Their algorithm is based on the suffix array and
the LCP array of $w$~\cite{manber93:_suffix}.

In this paper, we propose a faster LZ-LFS compression algorithm
which computes $(w', \Factors, F)$ in $O(n \log n)$ time with $O(n)$ space,
which is based on suffix trees
and Nakamura et al.'s algorithm~\cite{nakamura09:_linear_longest_first_}
for a grammar compression with LFS.

\subsection{Differences between LZ-LFS and grammar compression with LFS}

Here, we briefly describe main differences between
LZ-LFS and grammar compression with LFS.
In the sequel, grammar compression with LFS will simply be called LFS.

The biggest difference is that while the output of LFS is
a context free grammar that generates only the input string $w$,
that of LZ-LFS is not a grammar.
Namely, in LFS each selected occurrence of the LR is replaced with
a new non-terminal symbol,
but in LZ-LFS each selected occurrence of the LR is represented
as a pointer to the left-most occurrence of the LR in the current string $w_k$.
This also implies that in LZ-LFS the left-most occurrence of the LR
can remain in the string $w_{k+1}$ for the next $(k+1)$-th step.
On the other hand, in LFS no occurrences of the LR are left
in the string for the next step.

Because of Type 1 occurrences,
a repeat which only has overlapping occurrences in the current string $w_k$
can become an LR in LZ-LFS.
On the contrary, since LFS is a grammar-based compression,
LFS always chooses a longest repeat which has non-overlapping occurrences.

The above differences also affect technical details of the algorithms.
Nakamura et al.'s algorithm for LFS maintains
an incomplete version of the sparse suffix tree~\cite{KarkkainenU96}
of the current string.
On the other hand, our algorithm for LZ-LFS maintains
the suffix tree of the current string $w_k$ in each $k$-th step.

\subsection{On parameters $\alpha$ and $\beta$}

The algorithm of Mauer et al.~\cite{MauerBO17}
uses the suffix array and the LCP array~\cite{manber93:_suffix}
of the input string $w$,
and finds an LR $x_k$ for $w_k$ at each $k$-th step
using a maximal interval of the LCP array.

The suffix array $\SA$ for a string $w$ of length $n$ is a permutation
of $[1..n]$ such that $\SA[j] = i$ iff
$w[i..]$ is the lexicographically $j$-th suffix of $w$.
The LCP array $\LCP$ for $w$ is an array of length $n$
such that $\LCP[1] = 0$ and
$\LCP[i]$ stores the length
of the longest common prefix of $w[\SA[i-1]..]$ and $w[\SA[i]..]$
for $2 \leq i \leq n$.

For a positive integer $p$,
an interval $[i..j]$ of $\LCP$ array of $w$
is called a \emph{$p$-interval} if
(1) $\LCP[i-1] < p$,
(2) $\LCP[k] \geq p$ for all $i \leq k \leq j$,
(3) $\LCP[k] = p$ for some $i \leq k \leq j$, and
(4) $\LCP[j+1] < p$ or $j = n$.
An interval $[i..j]$ of $\LCP$ array of $w$
is called a \emph{maximal interval} if
it is a $p$-interval for some $p \geq 1$
and the longest common prefix of length $p$
for all the corresponding suffixes $w[\SA[i]..], \ldots, w[\SA[j]..]$
is a maximal repeat of $w$.
In each step of Mauer et al.'s method, the algorithm picks up a maximal interval
as a candidate for an LR to replace.

Let $\bit(w')$, $\bit(F)$, and $\bit(\Factors)$ respectively denote
the average number of bits to encode a single character from $w'$,
an element of $F$, and an element of $\Factors$ with a fixed encoding scheme.
The original algorithm by Mauer et al.~\cite{MauerBO17}
uses two parameters $\alpha$ and $\beta$ such that
$\alpha = \frac{\bit(\Factors)}{\bit(w')}$ and
$\beta = 1 + \frac{\bit(F)}{\bit(w')}$.
In each $k$-th step, their algorithm performs
replacement of an LR $x_k$ of length $\len_k$
only if the following conditions holds:
\begin{equation}
  \len_k \geq \frac{\alpha}{s} + \beta,
  \label{eqn:alpha_beta}
\end{equation}
where $s$ denotes the number of Type 2 or Type 3 occurrences
of the LR $x_k$ in the current string $w_k$.
However, since the values of $\alpha$ and $\beta$ cannot be precomputed,
in their implementation of LZ-LFS,
they use ad-hoc pre-determined values for $\alpha$ and $\beta$.
In particular, they set $\alpha = 30$ and $\beta = 80$
as default values
in their experiments (see~\cite{MauerBO17} for details).

However, we have found that there exist a series of strings
for which Mauer et al.'s algorithm fails to recursively replace LRs
for \emph{any} pre-determined values for $\alpha$ and $\beta$.

Consider a series of strings
\[
w = aXab_0 aXab_1 \cdots aXab_{s}\$,
\]
where $s \geq 1$, $a, b_1, \ldots, b_s \in \Sigma$,
$a \neq b_i$ for any $0 \leq i \leq s$,
$b_i \neq b_j$ for any $0 \leq i \neq j \leq s$,
and $X \in (\Sigma \setminus \{a, b_0, \ldots, b_s, \$\})^+$.
This string $w = w_1$ has a unique LR $aXa$.
Hence we have $\len_1 = r+2$, where $r = |X|$.
Since there are $s > 1$ non-overlapping occurrences of $aXa$
which do not overlap with the left most occurrence of $aXa$ in $w$,
those occurrences are of Type 3.
For this LR $aXa$ to be replaced with $\#_1$,
Equation~(\ref{eqn:alpha_beta}) or alternatively
$r \geq \frac{\alpha}{s} + \beta -2$ needs to hold.
Now let us choose $1 \leq |X| = r < \beta - 1$
and $s \geq \alpha$.
Then, since $\frac{\alpha}{s} \leq 1$,
Equation~(\ref{eqn:alpha_beta}) never holds for such $r$.
Hence, the original algorithm of Mauer et al. does not
replace $aXa$ and tries to find a next LR (which can be shorter than $aXa$).
In this case, the second longest repeats are
$aX$ and $Xa$ of length $r+1$ each.
However, since neither is $aX$ nor $Xa$ a maximal repeat of $w$,
it is not represented by a maximal interval of the LCP array.
Hence, neither is $aX$ nor $Xa$ selected for replacement.
Moreover, note that even $X$ is not a maximal repeat of $w$,
and that there are no repeats of length at least two
consisting only of $a$ and/or $b_i$~($0 \leq i \leq s$).
Therefore, Mauer et al.'s algorithm terminates at this point
and does not compress this string $w = aXab_0 aXab_1 \cdots aXab_{s}\$$ at all,
even though it is highly repetitive and contains quite long repeats
(e.g., for Mauer et al.'s default value $\beta = 80$, $X$ can be as long as $78$).

We also remark that one can easily construct instances
where more candidates of LRs have to be skipped,
by adding other strings in a similar way to $X$ into the string,
e.g.,
$aXab_0 aXab_1 \cdots aXab_{s} aYa c_0 aYa c_1 \cdots aYa c_s\$$,
and so on.

Given the above observation,
in our algorithm that follows,
we will omit the condition of Equation~\ref{eqn:alpha_beta},
and will replace Type 1, 2, 3 occurrences of \emph{any} selected LR.

\section{$O(n \log n)$-time algorithm for LZ-LFS}
\label{sec:nlogn_algorithm}

In this section, we show the following result:
\begin{theorem} \label{theo:main_result}
  Given a string $w$ of length $n$,
  our algorithm for LZ-LFS works in $O(n \log n)$ time with $O(n)$ space.
\end{theorem}

We begin with describing a sketch of our LZ-LFS algorithm.
Let $w$ be the input string of length $n$ and let $w_1 = w$.
As a preprocessing, we construct $\STree(w_1)$ in
$O(n \log \sigma)$ time and $O(n)$ space~\cite{Weiner,McC76,Ukk95},
where $\sigma \leq n$ is the number of distinct characters
that occur in $w$.

In the first step of the algorithm,
we find an LR $x_1$ of $w_1$ with the aid of $\STree(w_1)$.
Let $w_{k}$ denote the string in the $k$-th step of the algorithm.
For a technical reason,
when computing $w_{k+1}$ from $w_{k}$,
we use a special symbol $\#_k$ that does not occur in $w_{k}$,
and replace the selected occurrences of an LR $x_k$ in $w_{k}$ with $\#_k$.
The reason will become clear later.

For each $k$-th step,
we denote by $\len_k$ the length of an LR of $w_{k-1}$,
namely, $\len_k = |x_k|$.
At the end of each $k$-th step,
we update our tree so that it becomes identical to $\STree(w_{k+1})$,
so that we can find an LR $x_{k+1}$ for the next $(k+1)$-th step.


\subsection{How to find an LR $x_k$ using $\STree(w_{k})$}

Suppose that we maintain $\STree(w_{k})$ in each $k$-th step.
The two following lemmas are keys to our algorithm.
There, each $\#_k$ used at each $k$-th step
is regarded as a single character of length one,
rather than a representation of the LR of length $\len_k \geq 2$
that was replaced by $\#_k$.

\begin{lemma} \label{lem:no_hash_included}
  For each $k$-th step,
  let $v$ be any internal explicit node of $\STree(w_k)$
  of string depth at least two.
  Then, the string represented by $v$ does not contain
  $\#_j$ with any $1 \leq j < k$.
\end{lemma}

\begin{proof}
  Assume on the contrary that
  the string represented by $v$ contains $\#_j$ for some $1 \leq j < k$.
  Since $v$ is an internal explicit node of $\STree(w_k)$,
  $v$ occurs at least twice in $w_k$.
  Since $|v| \geq 2$,
  we have that $\len_k \geq |v| > \len_j$.
  However, this contradicts the longest first strategy
  such that $\len_j \geq \len_k$ must hold.
\end{proof}

\begin{lemma} \label{lem:LR_node_leaf_children}
  For each $k$-th step,
  any LR of $w_{k}$ is represented by an internal node
  of $\STree(w_{k})$.
\end{lemma}

\begin{proof}
  Suppose on the contrary that
  an LR $x$ of $w_{k}$ is represented by an implicit node
  of $\STree(w_{k})$,
  and let $(u, v)$ be the edge on which $x$ is represented.
  Note that $|v| > |x|$.
  Since $x$ is an LR, $x$ must occur at least twice in $w_{k}$
  and hence $v$ cannot be a leaf of $\STree(w_{k})$.
  This implies that $v$ is an internal branching node
  and hence $v$ occurs at least twice in $w_{k}$.
  However, this contradicts that $x$ is an LR of $w_k$.
\end{proof}

Based on Lemmas~\ref{lem:no_hash_included} and \ref{lem:LR_node_leaf_children},
we can find an LR at each step as follows.
In each $k$-th step of our algorithm,
we maintain an array $B_{k}$ of length $n$ such that
$B_{k}[l]$ stores a list of all explicit internal
nodes of string depth $l$ that exist in $\STree(w_{k})$.
Hence, $B_{k}[\len_k]$ will be the leftmost entry of $B_k$
that stores a non-empty list of existing nodes.
We do not store nodes of string depth one.
Any node of string depth one represents
either a single character from the original string $w$
or $\#_j$ for some $1 \leq j < k$ which will never be replaced
in the following steps.
Therefore, $B_{k}[1]$ is always empty at every $k$-th step.

The initial array $B_1$ can easily be computed in $O(n)$ time
by a standard traversal on $\STree(w_1) = \STree(w)$.
We can also compute in $O(n)$ time the length $\len_1$ of an LR for $B_1$
in a na\"ive manner.
We then pick up the first element in the list stored at $B_1[\len_1]$
as an LR $x_1$ of $w_1$ to be replaced with $\#_1$.
After the replacement, we remove $x_1$ from the list,
and proceed to the next step.
In the next subsection, we will show how to efficiently
update $B_{k}$ to $B_{k+1}$.

The algorithm terminates when the string
contains no repeats of length at least two.
Let $w_m$ denote this string,
namely, the algorithm terminates at the $m$-th step.
In this last $m$-th step,
$\STree(w_m)$ consists only of the root, the leaves,
and possibly internal explicit nodes of string depth one.
See also an example in Appendix.

In the next subsection, we will show how to efficiently
update $\STree(w_k)$ to $\STree(w_{k+1})$
and $B_k$ to $B_{k+1}$ in a total of $O(n)$ time
for all $k = 1, \ldots, m-1$.
We also remark that $m$ cannot exceed $n/2$ since at least
two positions are taken by the replacement of an LR at each step.

Now, let us focus on
how our algorithm works at each $k$-th step.
The next lemma shows how we can find the occurrences of
an LR of each step efficiently.

\begin{lemma}
  Given a node of $\STree(w_{k})$ which represents an LR $x_k$
  of $w_{k}$ at each $k$-th step,
  we can compute Type 1, 2, 3 occurrences of $x_k$ in $w_{k}$
  in a total of $O(n \log n)$ time and $O(n)$ space for all steps.
\end{lemma}

\begin{proof}
  It follows from Lemma~\ref{lem:LR_node_leaf_children} that
  all children of the node for $x_k$ are leaves in $\STree(w_{k})$.
  We sort all the leaves in increasing order of their id's
  (i.e., the beginning positions of the corresponding suffixes).
  If $d_k$ is the number of the above-mentioned leaves,
  then this can be done in $O(d_k \log d_k)$ time and $O(d_k)$ space
  by a standard sorting algorithm.
  It is clear that we can compute Type 1, 2, and/or 3 occurrences
  of $x_k$ in $w_{k}$ from this sorted list, in $O(d_k)$ time.

  Each occurrence $i$ of $x_k$ but the leftmost one either
  (a) is replaced with $\#_k$, or
  (b) overlaps with another occurrence of $x_k$ that is replaced with $\#_k$.
  In case (a), it is guaranteed that
  there will be no LRs that begin at position $i$ in the following steps,
  since LZ-LFS chooses repeats in a longest first manner.
  In case (b),
  there is another occurrence $j$ of $x_k$
  that is replaced with $\#_k$ and  $i \in [j+1..j+\len_k-1]$.
  Since these positions in this range $[j+1..j+\len_k-1]$
  are already taken by the replacement of $x_k$ with $\#_k$,
  there will be no LRs that begin at position $i$ in the following steps.
  One delicacy is the leftmost occurrence $\ell$ of $x_k$,
  since the corresponding interval $[\ell..\ell+\len_k-1]$ can contain
  up to $\len_k$ occurrences of $x_k$,
  and these positions may retain the original characters in
  the string $w_{k+1}$ for the next $(k+1)$-th step.
  However, since at least one occurrence of $x_k$ is always replaced,
  the cost of sorting the leaves whose id's are in range
  $[\ell..\ell+\len_k-1]$ can be charged
  to an occurrence of $x_k$ that is replaced with $\#_k$.

  Overall, the time cost to sort all $d_k$ children of $x_k$
  can be charged to the intervals of the occurrences of $x_k$ in $w_{k}$
  that are replaced with $\#_k$'s.
  Therefore, the total time cost for sorting the corresponding leaves
  in all $m$ steps is $O(\sum_{k=1}^m (d_k \log d_k)) = O(n \log n)$,
  where the equality comes from the fact that
  $\sum_{k=1}^m d_k = O(n)$ and $d_k \leq n$ for each $k$.

  The space complexity is clearly $O(n)$.
\end{proof}

\subsection{How to update $\STree(w_{k})$ to $\STree(w_{k+1})$}

In this subsection,
we show how to update $\STree(w_k)$ to $\STree(w_{k+1})$.

Let $i$ be any occurrence (Type 1, 2, or 3) of an LR $x_k$ in $w_{k}$
which will be replaced with $\#_k$ in the $k$-th step.
Since $|x_k| = \len_k \geq 2$,
the replacement with $\#_k$ will always shrink the string length.
However, it is too costly to relabel the integer pairs
for the suffix tree edge labels with the positions in the shrunken string.
To avoid this, we suppose that each selected occurrence
of $x_k$ is replaced with $\#_k \bullet^{\len_k-1}$,
where $\bullet$ is a special symbol that does not occur in the original string $w$.
Namely, $\#_k$ is now at position $i$
and positions $i+1, \ldots, i+\len_k-1$ are padded with $\bullet$'s.
This ensures that the length of $w_{k}$ remains $n$ for each $k$-th step,
and makes it easy for us to design our LZ-LFS algorithm.

If an occurrence of $x_k$ at position $i$ is replaced with $\#_k$,
then the positions in range $[i+1..i+\len_k-1]$ are taken away from the string.
This range $[i+1..i+\len_k-1]$ is therefore not considered in the following steps,
and is called a \emph{dead zone}.
Also, since any LRs in the following steps are of length at most $\len_k$,
it suffices for us only to take care of the
substrings in range $[i-\len_k,..i]$.
This range is called as an \emph{affected zone}.
See Figure~\ref{fig:dead_affected_zones} for illustration
of a dead zone and affected zone.

\begin{figure}[h]
  \centerline{
    \includegraphics[scale=0.6]{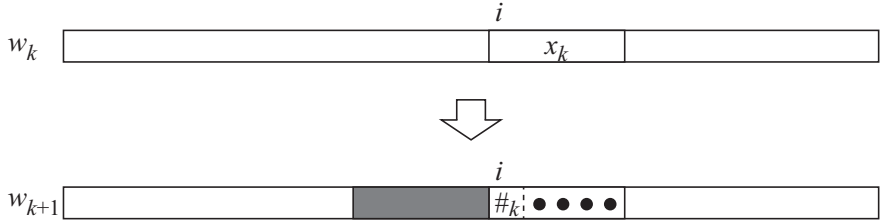}
  }
  \caption{An occurrence of LR $x_k$ at position $i$ in the current string $w_k$ is replaced with $\#_k$. In the next string $w_{k+1}$, the range padded with $\bullet$'s is the dead zone and the gray range is the affected zone for this occurrence of $x_k$ at position $i$.}
  \label{fig:dead_affected_zones}
\end{figure}

In our suffix tree update algorithm,
we will remove the leaves for the suffixes that begin in the dead zones,
and modify the leaves for the suffixes that begin in the affected zones.

Let $q_k$ denote the number of selected occurrences (Type 1, 2, or 3)
of $x_k$ in $w_{k}$ to be replaced with $\#_k$.
We will replace the selected occurrences of $x_k$ from left to right.
For each $1 \leq h \leq q_k$,
let $i_h$ denote the $h$-th selected occurrence of $x_k$ from the left,
and let $w_{k}^h$ denote the string where
the $h$ occurrences $i_1, \ldots, i_h$ of $x_k$ from the left
are already replaced with $\#_k$'s.
Namely, $w_k^0 = w_k$ and $w_{k}^{q_k} = w_{k+1}$.

Suppose that we have processed the $h-1$ occurrences of $x_k$ from the left,
and we are to process the $h$-th occurrence $i_h$ of $x_k$.
Namely, we have maintained $\STree(w_k^{h-1})$ and we are to update
it to $\STree(w_k^{h})$.

\subsubsection{How to process the dead zones.}
First, we consider how to deal with the dead zone
$[i_h+1..i_h+\len_k-1]$ for this occurrence $i_h$ of $x_k$ in $w_k^{h-1}$.
Since the positions in the dead zone
will not exist in the modified string,
and since no substrings beginning in this dead zone
can be an LR in the following steps,
we remove the leaves for the suffixes that begin at
the positions in the dead zone $[i_h+1..i_h+\len_k-1]$.
In case $i_h+\len_k-1 > n$,
which can happen only when $h = q_k$,
then the dead zone for this occurrence is $[i_h+1..n]$.
In any case, we can easily remove those leaves
in linear time in the number of the removed leaves.

\subsubsection{How to process the affected zones.}
Next, we consider how to deal with the affected zone
$[i_h-\len_k..i_h]$ for this occurrence $i_h$ of LR $x_k$ in $w_k^{h-1}$.
Let $y = w_k^{h-1}[i_h-\len_k..i_h-1]$,
namely, $y$ is the left context of length $\len_k$
from the occurrence of $x_k$ at position $i_h$.
Let $y'$ be the longest non-empty suffix of $y$
such that $x_k$ down the locus of $y'$ spans more than one edge in the tree.
If such a node does not exist, then let $y' = \varepsilon$.
For each suffix of $y$ that is longer than $y'$,
$x_k$ down its locus is represented on a single edge.
Hence, it is ``automatically'' be replaced with $\#_k$
by replacing the occurrence of $x_k$ at position $i_h$ in the current string
$w_k^{h-1}$ with $\#_k \bullet^{\len_k-1}$.
Therefore, no explicit maintenance on the tree topology is needed
for these suffixes of $y$.

Now we consider the suffixes $y_j = y[j..\len_k-1]$ of $y$
that are not longer than $y'$, where $j = \len_k-|y'|+1, \ldots, \len_k-1$.
Now $x_k$ down the locus of each $y_j$
spans more than one edge,
and it will have to be replaced with a (single) special symbol $\#_k$.
This introduces some changes in the tree topology.
We note that the locus of $y_j x_k$ in the suffix tree
before the update is on the edge that leads to the leaf with id $i_h-|y_j|$,
since otherwise $y_j x_k$ must occur twice in the string,
which contradicts our longest first strategy.
Thus, we re-direct the edge
that leads to the leaf with id $i_h-|y_j|$ from its original parent
to the node that represents $y_j$
(if it is an implicit node, then we create a new explicit node there).
See Figure~\ref{fig:leaf_edge_redirection} for illustration.
This event can also be found at the first step of our concrete example
shown in Appendix.

\begin{figure}[h]
  \centerline{
    \includegraphics[scale=0.6]{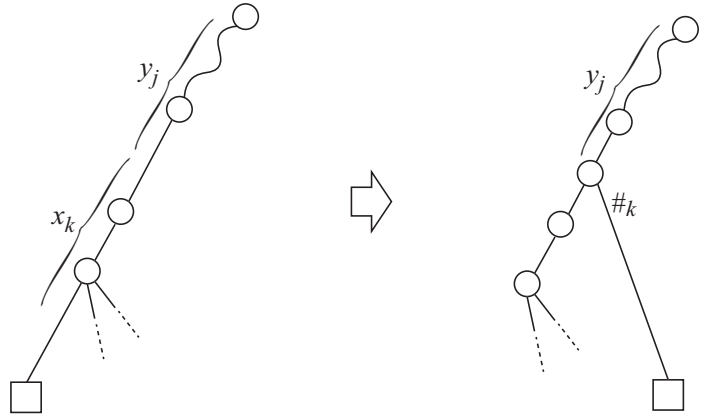}
  }
  \caption{Illustration for a leaf edge redirection,
    where the circles represent internal explicit nodes and the
    square represents the leaf with id $i_h-|y_j|$.
    Since $x_k$ down the locus of $y_j$ spans more than one edge,
    the leaf edge is redirected from its original parent to $y_j$.
    This figure shows the case where a new internal node for $y_j$ is created.}
  \label{fig:leaf_edge_redirection}
\end{figure}

The remaining problem is how to find the loci for
the suffixes of $y$ in the tree.
We find them in decreasing order of their length.
For the first suffix $y[1..\len_k] = y$,
we find the locus of $y$
by simply traversing $y$ from the root of the suffix tree.
There are two cases to consider:
\begin{enumerate}
\item[(A)]
If this locus for $y_1 = y$
is an explicit node in $\STree(w_{k}^{h-1})$,
then by the property of the suffix tree,
all suffixes of $y$ are also represented by explicit nodes.
Hence, we can find the loci for all the suffixes
using a chain of suffix links from node $y$ down to the root.

\item[(B)]
If this locus for $y_1 = y$
is an implicit node in $\STree(w_{k}^{h-1})$,
then we use the suffix link of the parent $u_{1}$ of $y_{1}$.
Let $u'_2 = \slink(u_1)$.
We go downward from $u'_2$ until finding
the deepest node $u_2$ whose string depth is not greater than
$|y_2| = \len_k-1$.
If the string depth $u_2$ equals $|y_2|$ (i.e. $|u_2| = |y_2|$),
then the locus of $y_2$ is on an explicit node.
Hence, we can continue with $y_3$ as in Case (A) above.
Otherwise (if $|u_2| < |y_2|$),
then the locus of $y_2$ is on an out-going edge of $u_2$.
We then continue with $y_3$ in the same way as for $y_2$.
\end{enumerate}

Suppose we have processed all the $q_k$ selected occurrences
of $x_k$ in $w_k$.
The next lemma guarantees that
re-direction of the leaf edges do not break
the property of the suffix tree.

\begin{lemma} \label{lem:distinct_child_labels}
Let $v$ be any non-root internal explicit node of the the tree
obtained by updating $\STree(w_k^{h-1})$ as above.
Then, the labels of the out-going edges of $v$
begin with mutually distinct characters.
\end{lemma}

\begin{proof}
  Notice that in each $k$-th step,
  the label of any re-directed edge begins with $\#_k$.
  Since $\#_k \neq \#_j$ for any $1 \leq j < k$
  and $\#_k$ does not occur in $w_k$,
  it suffices for us to show that
  there is at most one out-going edge of $v$
  whose label begins with $\#_k$.

  If there are two out-going edges of $v$
  whose labels begin with $\#_k$,
  then there are at least two leaves whose path label begin with $v \#_k$.
  Thus $v \#_k$ occurs in $w_k$ at least twice.
  Since $v$ is not a root, $|v| \geq 1$.
  If $x_k$ is the LR that was replaced by $\#_k$,
  then $|v x_k| > |x_k| = \len_k$, which contradicts that $x_k$
  was an LR at the $k$-th step.

  Thus, the labels of out-going edge of any node $v$
  begin with mutually distinct characters.
\end{proof}

The root of the resulting tree
has a new child which represents $\#_k$,
and the children of this new node are the leaves
that correspond to the occurrences of the LR
that have been replaced by $\#_k$.

Notice that the affected zone $[i_h-\len_k..i_h-1]$
for the occurrence $i_h$ may overlap with
the dead zone $[i_{h-1}+1..i_{h-1}+\len_k-1]$ for the
previous occurrence $i_{h-1}$.
In this case, the affected zone for $i_h$
is trimmed to $[i_{h-1}+\len_k..i_h-1]$ and we perform
the same procedure as above for this trimmed affected zone.

\begin{lemma}
  Our algorithm updates $\STree(w_k)$ to $\STree(w_{k+1})$
  for every $k$-th step in a total of $O(n \log \sigma)$ time
  with $O(n)$ space.
\end{lemma}

\begin{proof}
  First, let us confirm the correctness of our algorithm.
  It follows from Lemma~\ref{lem:LR_node_leaf_children} that
  in each $k$-th step the new internal explicit nodes
  that are created in this step
  can have string depth at most $\len_k$.
  Therefore, in terms of updating $\STree(w_k)$ to $\STree(w_{k+1})$,
  it suffices for us to consider
  only the affected zone for each occurrence of LR $x_k$.
  Lemma~\ref{lem:distinct_child_labels} guarantees
  that the label of the out-going edges of the same node
  begin with mutually distinct characters.
  It is clear that the leaves for the suffixes which begin in the dead zones
  have to be removed, and only those leaves are removed.
  Thus, our algorithm correctly updates $\STree(w_k)$ to $\STree(w_{k+1})$.

  Second, let us analyze the time complexity of our algorithm.
  For each occurrence $i_h$ of $x_k$,
  finding the locus for the first suffix $y = w_k^{h-1}[i_h-\len_k..i_h-1]$
  takes $O(\len_k \log \sigma)$ time.
  Then, the worst case scenario is that
  Case (B) happens for all $\len_k$ suffixes of $y$.
  For each shorter suffix $y[i..\len_k]$ with $i = 2, \ldots, \len_k$,
  the above algorithm traverses at most
  $|u_{j}|-|u'_{j}| = |u_j|-|\slink(u_{j-1})| = |u_j|-|u_{j-1}|+1$ edges.
  Hence, for all the shorter suffixes of $y$,
  the number of edges traversed is bounded by
  $\sum_{j=2}^{\len_k}(|u_j|-|u_{j-1}|+1) =
  |u_{\len_k}|-|u_1|+\len_k-1 < 2 \len_k$.
  Hence, finding the locus for the shorter suffixes of $y$
  also takes $O(\len_k \log \sigma)$ time.
  The $\len_k$ term in the $O(\len_k \log \sigma)$ complexity
  can be charged to each selected occurrence of LR $x_k$,
  which is replaced with $\#_k \bullet^{\len_k-1}$.
  Therefore, the total time cost
  to update the suffix tree for all steps is $O(n \log \sigma)$.
  The space usage is clearly $O(n)$.
\end{proof}

\subsection{How to update $B_k$ to $B_{k+1}$}

Suppose we have $B_{k}$ in the $k$-th step,
and we would like to update it to $B_{k+1}$ for the next $(k+1)$-th step.
Let $u$ be an internal branching node of $\STree(w_{k-1})$
that is to be removed in $\STree(w_k)$.
This can happen when $u$ has only two children,
one of which is a leaf to be removed from the current suffix tree.
We then remove $u$ from the list stored in $B_{k-1}[|u|]$,
and connect its left and right neighbors in the list.

When we replace an LR $x_k$ of $w_{k}$
with $\#_k \bullet^{\len_k-1}$,
an implicit node $v$ of $\STree(w_{k})$
may become branching due to the new symbol $\#_k$
and hence a new explicit internal node for $v$
needs to be created to the suffix tree.
In this case, we add this new node for $v$ at the end of
the list stored in $B_{k}[|v|]$.
After these procedures are performed for all such nodes,
we obtain $B_{k+1}$ for the next $(k+1)$-th step.

\begin{lemma}
  At every $k$-th step,
  we can update $B_k$ and maintain $\len_k$ in a total of $O(n)$ time and space.
\end{lemma}

\begin{proof}
  Initially, at most $n-1$ internal nodes are stored in $B_1$.
  Also, the total number of newly created nodes
  is bounded by the total size of the affected zones
  for the replaced occurrences of the LRs in all the steps,
  which can be charged to the positions that are taken
  by replacement of LRs for all the steps.
  As was shown in the previous subsection,
  once a position in the original string
  is taken by replacement of an LR,
  then this position will never be considered in the following steps.
  Thus, the total number of newly created nodes is bounded by $n$.
  Clearly, computing the initial array $B_1$ from $\STree(w_1)$
  takes $O(n)$ time,
  and deletion and insertion of a node
  on a list stored at an entry of $B_k$ takes $O(1)$ time each
  (we use doubly linked lists here).

  It follows from Lemma~\ref{lem:LR_node_leaf_children}
  and our suffix tree update algorithm that
  at each $k$-th step
  any newly created node has string depth at most $\len_k$,
  and $\len_k$ is monotonically non-increasing as $k$ grows.
  Hence, we can easily keep track of $\len_k$ for all steps
  in a total of $O(n)$ time.

  The space usage is clearly $O(n)$.
  \qed
\end{proof}

After computing $w_m$ for the final $m$-th step,
we replace every $\#_k$ in $w_m$ with $\#$ for every $k$,
and obtain the final string $w'$ for LZ-LFS.

Summing up all the discussions above,
we have proved our main result in Theorem~\ref{theo:main_result}.

In Appendix,
we show a concrete example on how
our suffix-tree based LZ-LFS algorithm works.

\section{$O(n)$-time algorithm for simplified LZ-LFS}

In this section, we show that
a simplified version of LZ-LFS can be computed in
$O(n)$ time and space,
by a slight modification to our $O(n \log n)$-time LZ-LFS algorithm
from Section~\ref{sec:nlogn_algorithm}.

By a ``simplified version'' of LZ-LFS,
we mean a variant of LZ-LFS where
Type 3 non-overlapping occurrences of an LR of each step can
be selected arbitrarily (namely, not necessarily in a left-greedy manner).
More formally, in our simplified version of LZ-LFS,
an occurrence $i$ of $x$ in $w$ is said to be of
Type 1/2 if the corresponding condition
as in Section~\ref{sec:LZ-LFS_prelim} holds, and
\begin{itemize}
\item Type 3 if $i$ is an occurrence of $x$ after $e$
  which is not of Type 2,
\end{itemize}
where $e$ is as defined in Equation~(\ref{eqn:e}).

Notice that there can be multiple choices for non-overlapping
Type 3 occurrences of LR $x_k$ in $w_k$ at each $k$-th step.
Our algorithm takes a maximal set of
non-overlapping Type 3 occurrences of $x_k$ in $w_k$ at each step,
so that no Type 3 occurrences remain in the string.
We remark that it is easy to compute
a maximal set of size at least
$\max\{\lceil |\LG_{w_k[e+1..]}(x_k)|/2 \rceil, 2\}$,
namely, this strategy allows us to select at least
half the number of left-greedily selected Type 3 occurrences.
Since this does not require to sort the occurrences of $x_k$,
we can perform all the steps in a total of $O(n)$ time,
as follows:

\begin{theorem}
  Given a string $w$ of length $n$ over an integer alphabet of size $n^{O(1)}$,
  our algorithm for a simplified version of LZ-LFS
  works in $O(n)$ time and space.
\end{theorem}

\begin{proof}
  As a preprocessing,
  we build $\STree(w)$ in $O(n)$ time and space~\cite{Farach-ColtonFM00}.

  We use essentially the same approach as in the previous section.
  Namely, we maintain the suffix tree for each step of our algorithm,
  and find Type 1, 2, and/or 3 occurrences of a selected LR using
  the suffix tree that we maintain.

  Suppose that we are given a node $v$ that represents
  an LR $x_k$ in $w_k$ at the $k$-th step.
  Since all children of $v$ are leaves,
  we can easily compute the Type 1 occurrence of $x_k$ (if it exists)
  by a simple scan over the children's leaf id's.
  After this, by another simple scan,
  we can also compute the Type 2 occurrence of $x_k$ (if it exists).
  Then, we exclude the Type 1 and Type 2 occurrences,
  and any occurrences that overlap with the Type 1 and/or Type 2 occurrences,
  by removing the corresponding leaves which are children of $v$.
  We then select a maximal set of non-overlapping Type 3 occurrences of $x_k$
  by picking up a child of $v$ in an arbitrary order,
  and choosing it if it does not overlap with any already-selected occurrences.

  Let $d_k$ be the number of children of $v$.
  As in the standard LZ-LFS,
  each position of the original string can be involved in
  at most one event of the replacement of an LR.
  Hence, each step of the above algorithm takes $O(d_k)$ time,
  and thus the total time complexity for all the steps of this algorithm is
  $O(\sum_{k=1}^{m} d_k) = O(n)$,
  where $m$ is the final step.

  The space complexity is clearly $O(n)$.
\end{proof}

\section{Conclusions and further work}

LZ-LFS~\cite{MauerBO17} is a new text compression method
that has both features of Lempel-Ziv 77~\cite{LZ77} and
grammar compression with longest first substitution~\cite{nakamura09:_linear_longest_first_}.

In this paper, we proposed a suffix-tree based algorithm
for LZ-LFS that runs in $O(n \log n)$ time and $O(n)$ space,
where $n$ denotes the length of the input string to compress.
This improves on Mauer et al.'s suffix-array based algorithm
that requires $O(n^2)$ time and $O(n)$ space.
We also showed that a simplified version of LZ-LFS,
where Type 3 occurrences may not be selected in a left-greedy manner,
can be computed in $O(n)$ time and space
with slight modifications to our LZ-LFS algorithm.

There are interesting open questions with LZ-LFS,
including:
\begin{enumerate}
\item Does there exist a linear $O(n)$-time algorithm for
  (non-simplified) LZ-LFS?
  The difficulty here is to select Type 3 occurrences of each selected LR
  in a left-greedy manner.
  We remark that Nakamura et al.'s linear $O(n)$-time algorithm~\cite{nakamura09:_linear_longest_first_}
  for grammar compression with LFS does \emph{not} always
  replace the left-greedy occurrences of each selected LR, either.
  Or, do there exist $\Omega(n \log n)$ lower bounds,
  probably by a reduction from sorting?

\item Does there exist a suffix-array based algorithm
  for LZ-LFS which works in time faster than $O(n^2)$?
  This kind of algorithm could be of practical significance.
\end{enumerate}

\bibliographystyle{abbrv}
\bibliography{ref}

\begin{thebibliography}{10}

\bibitem{Farach-ColtonFM00}
M.~Farach{-}Colton, P.~Ferragina, and S.~Muthukrishnan.
\newblock On the sorting-complexity of suffix tree construction.
\newblock {\em J. {ACM}}, 47(6):987--1011, 2000.

\bibitem{KarkkainenU96}
J.~K{\"{a}}rkk{\"{a}}inen and E.~Ukkonen.
\newblock Sparse suffix trees.
\newblock In {\em Proc. COCOON 1996}, pages 219--230, 1996.

\bibitem{LarssonDCC99}
N.~J. Larsson and A.~Moffat.
\newblock Offline dictionary-based compression.
\newblock In {\em DCC 1999}, pages 296--305, 1999.

\bibitem{manber93:_suffix}
U.~Manber and G.~Myers.
\newblock Suffix arrays: A new method for on-line string searches.
\newblock {\em SIAM J.~Computing}, 22(5):935--948, 1993.

\bibitem{MauerBO17}
M.~Mauer, T.~Beller, and E.~Ohlebusch.
\newblock A {L}empel-{Z}iv-style compression method for repetitive texts.
\newblock In {\em Proc. PSC 2017}, pages 96--107, 2017.

\bibitem{McC76}
E.~M. McCreight.
\newblock A space-economical suffix tree construction algorithm.
\newblock {\em J. ACM}, 23(2):262--272, 1976.

\bibitem{nakamura09:_linear_longest_first_}
R.~Nakamura, S.~Inenaga, H.~Bannai, T.~Funamoto, M.~Takeda, and A.~Shinohara.
\newblock Linear-time off-line text compression by longest-first substitution.
\newblock {\em Algorithms}, 2(4):1429--1448, 2009.

\bibitem{Storer77}
J.~Storer.
\newblock {NP}-completeness results concerning data compression.
\newblock Technical Report 234, Department of Electrical Engineering and
  Computer Science, Princeton University, 1977.

\bibitem{LZSS}
J.~Storer and T.~Szymanski.
\newblock Data compression via textual substitution.
\newblock {\em J. ACM}, 29(4):928--951, 1982.

\bibitem{Ukk95}
E.~Ukkonen.
\newblock On-line construction of suffix trees.
\newblock {\em Algorithmica}, 14(3):249--260, 1995.

\bibitem{Weiner}
P.~Weiner.
\newblock Linear pattern-matching algorithms.
\newblock In {\em Proc. of 14th IEEE Ann. Symp. on Switching and Automata
  Theory}, pages 1--11, 1973.

\bibitem{LZ77}
J.~Ziv and A.~Lempel.
\newblock A universal algorithm for sequential data compression.
\newblock {\em IEEE Transactions on Information Theory}, IT-23(3):337--343,
  1977.

\end{thebibliography}

\clearpage
\appendix
\section{Appendix: Example}

Here we show a concrete example on
how our suffix-tree based LZ-LFS algorithm works.
Consider input string
\[
  w = w_1 = \mathtt{abbaaccabccbaabcb\$}.
\]
We preprocess $w_1$ and build $\STree(w_1)$,
which is shown below.

\begin{figure}[h]
  \centerline{
    \includegraphics[scale = 0.4]{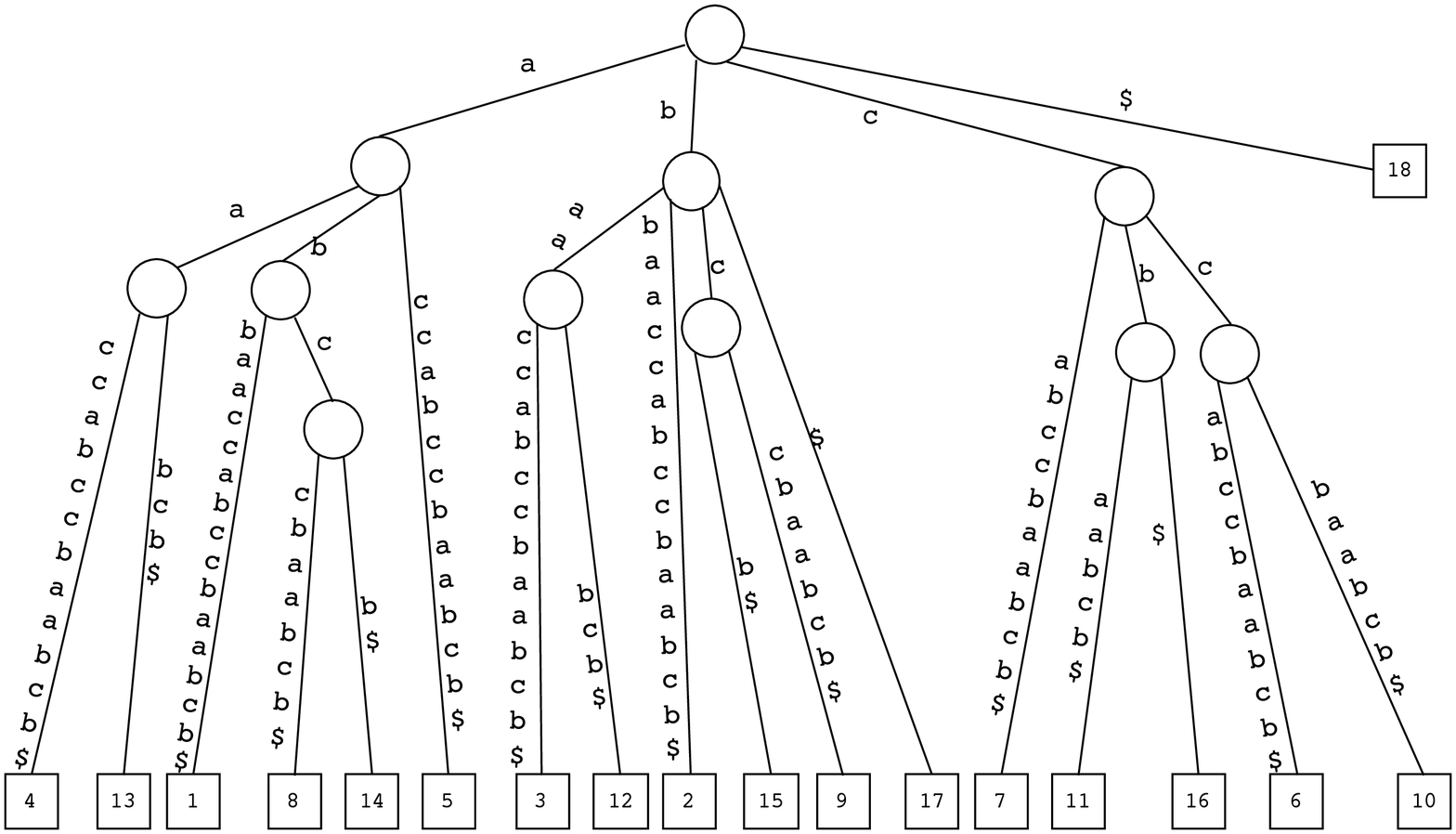}
  }
\end{figure}

Now we go on to the first step.
String $w_1$ has a unique LR $\mathtt{baa}$,
which occurs at positions $3$ and $12$.
The occurrence of $\mathtt{baa}$ at positoin 12 is replaced,
and the resulting string will be
\[
 w_2 = \mathtt{abbaaccabcc\#_1 \! \bullet \! \bullet bcb\$}.
\]

To update the suffix tree,
we remove the leaves with id $13$ and $14$ which are
in the dead zone $[13..14]$.
For simplicity, we omit any child of the root
which represents $\#_k$ for each $k$-th step.
The current tree is shown below.

\begin{figure}[h!]
  \centerline{
    \includegraphics[scale = 0.4]{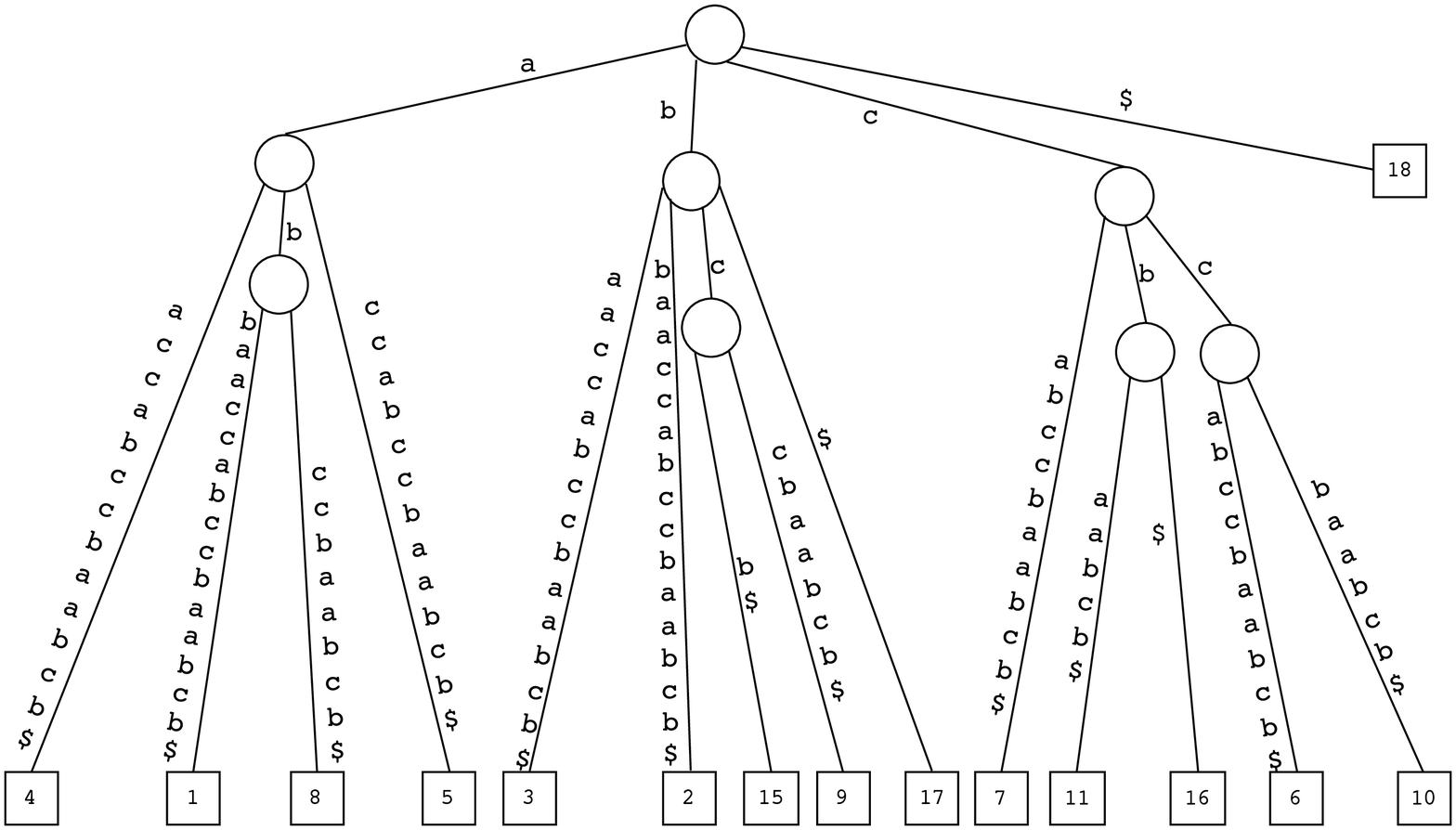}
  }
\end{figure}

Now we take care of the affected zone $[9..11]$
whose corresponding substring is $\mathtt{bcc}$.
We first find the locus of $\mathtt{bcc}$ by traversing the tree from the root.
Since $\mathtt{baa}$ down the locus of $\mathtt{bcc}$
is on an edge, no explicit maintainance is needed.
We then move to the locus of $\mathtt{cc}$
by first moving to node $\mathtt{c}$
using the suffix link of node $\mathtt{bc}$,
and reading the second $\mathtt{c}$ from there.
Again, $\mathtt{baa}$ is on an edge,
and hence no explicit maintainance is needed.
We then move to the locus of $\mathtt{c}$
by using the suffix link of $\mathtt{cc}$.
This time, $\mathtt{baa}$ spans more than one edge.
Hence, the leaf with id $11$ is redirected
from its current parent $\mathtt{cb}$ to its new parent $\mathtt{c}$.
The resulting tree $\STree(w_2) $ is the following.

\begin{figure}[h!]
  \centerline{
    \includegraphics[scale = 0.4]{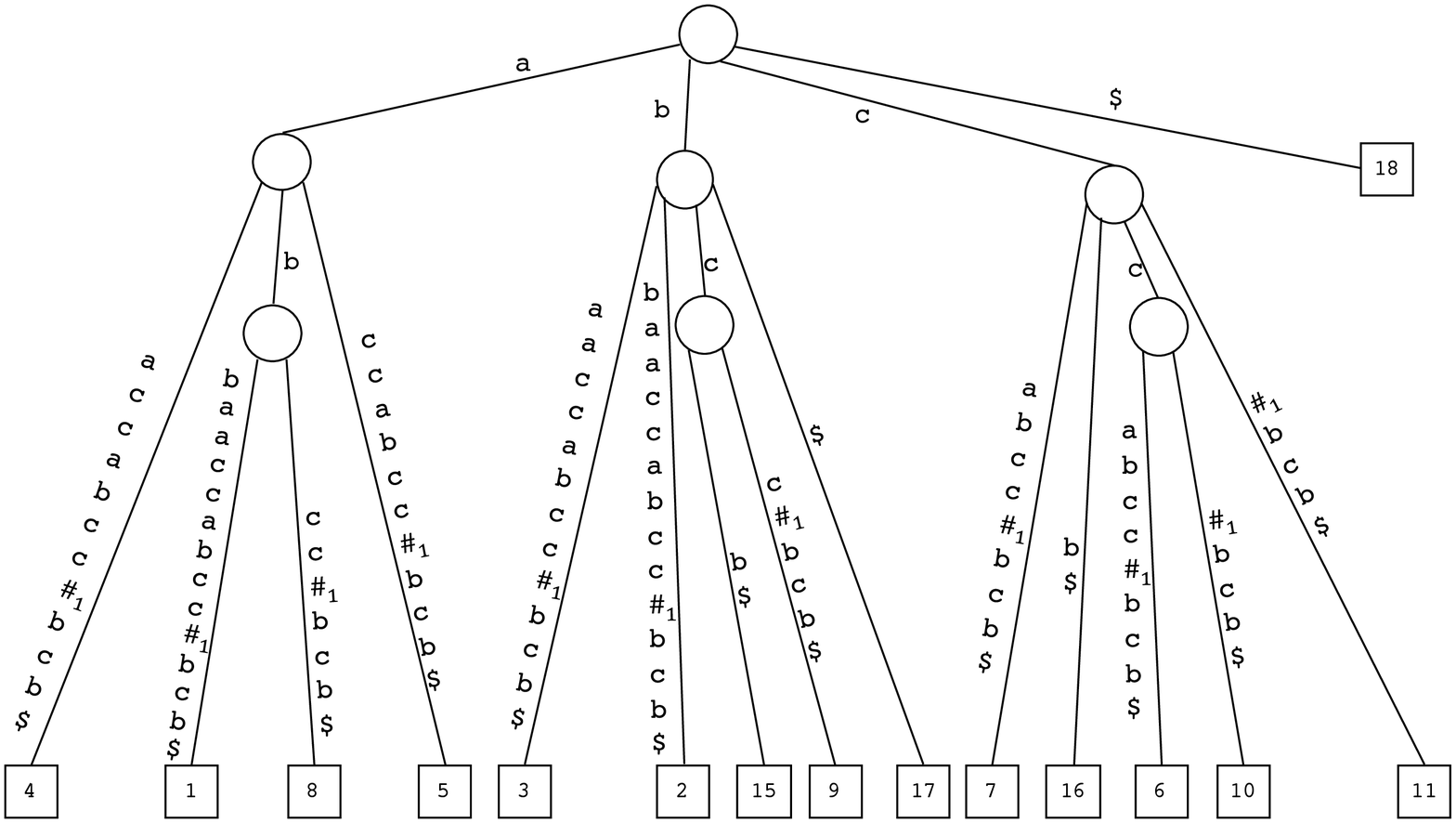}
  }
\end{figure}

Now we go on to the second step.
String $w_2$ has two LRs $\mathtt{ab}$ and $\mathtt{cc}$,
and let us choose $\mathtt{ab}$ for this second step.
This LR $\mathtt{ab}$ occurs at positions $1$ and $8$,
and $\mathtt{ab}$ occurring at position $8$ is replaced.
The resulting string will be
\[
  w_3 = \mathtt{abbaacc\#_2 \! \bullet \! cc\#_1 \! \bullet \! \bullet bcb\$}.
\]

To update the suffix tree,
we first remove the leaf for the dead zone $[9..9]$,
and the resulting tree is the following.

\begin{figure}[h!]
  \centerline{
    \includegraphics[scale = 0.4]{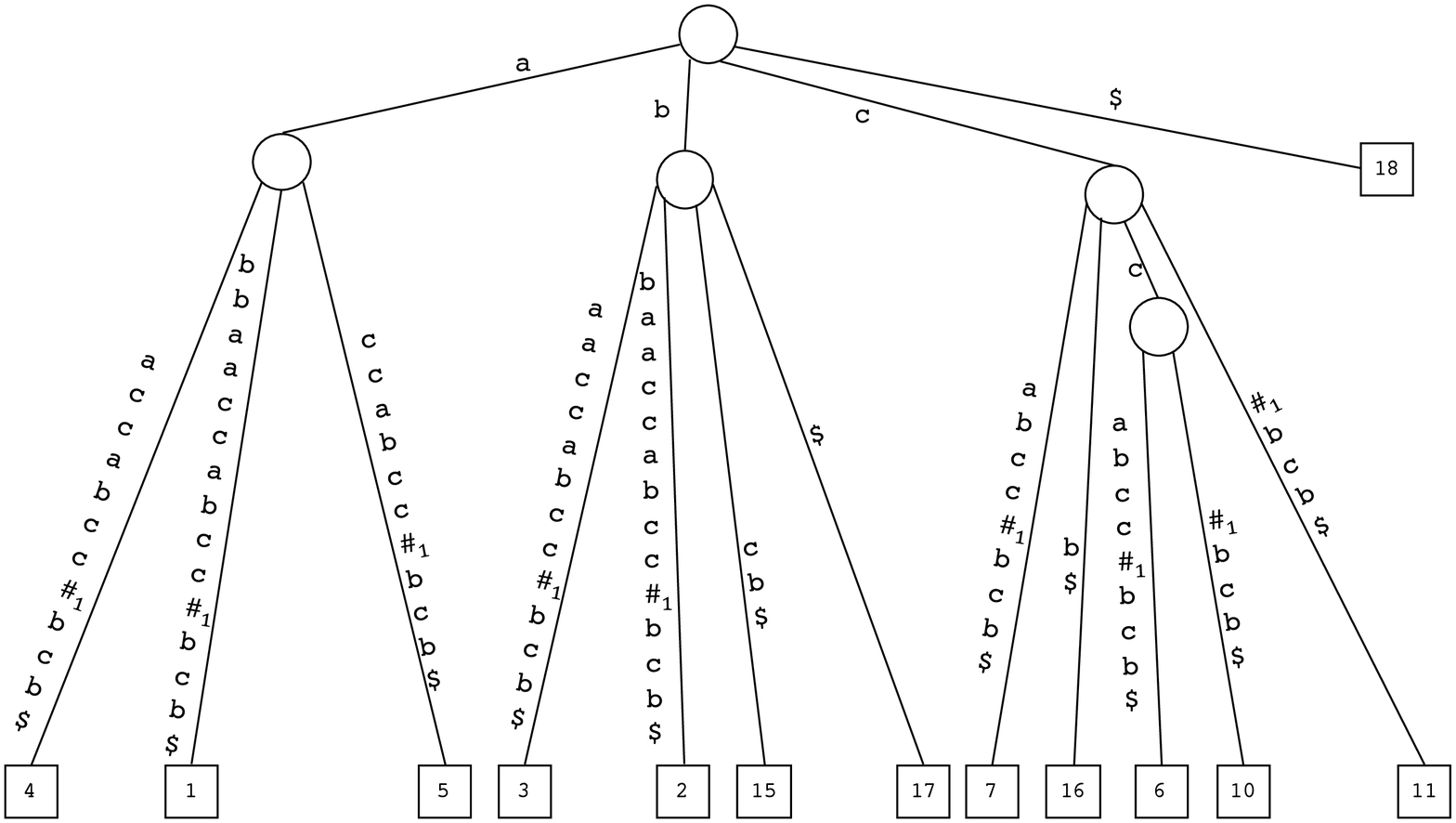}
  }
\end{figure}

Now we take care of the affected zone $[6..7]$
whose corresponding substring is $\mathtt{cc}$,
and we obtain $\STree(w_3)$ shown below.

\clearpage

\begin{figure}[h!]
  \centerline{
    \includegraphics[scale = 0.4]{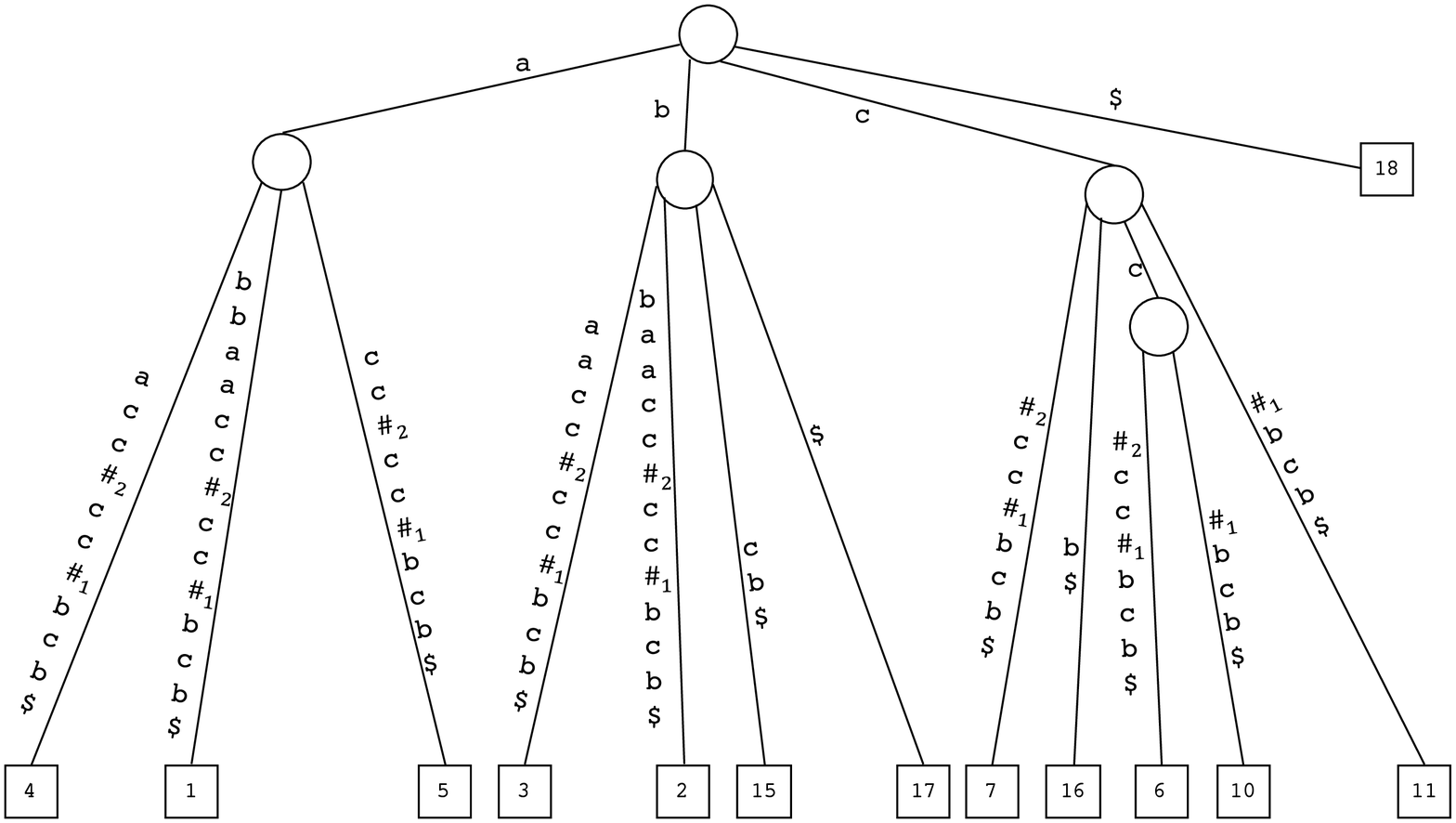}
  }
\end{figure}

Here, we remark that nodes $\mathtt{cc}$ and $\mathtt{c}$
have two out-going edges which begin with $\#_1$ and $\#_2$.
This is the reason why we use a distinct special symbol $\#_k$
for each $k$-th step.

Now we go on to the third step.
String $w_3$ has a unique LR $\mathtt{cc}$,
which occurs at positions $6$ and $10$.
The resulting string will be
\[
 w_4 = \mathtt{abbaacc\#_2 \! \bullet\! \#_3 \! \bullet \! \#_1 \! \bullet \! \bullet bcb\$}.
\]

After removing the leaf for the dead zone $[11.11]$,
we obtain the tree shown below.

\begin{figure}[h!]
  \centerline{
    \includegraphics[scale = 0.5]{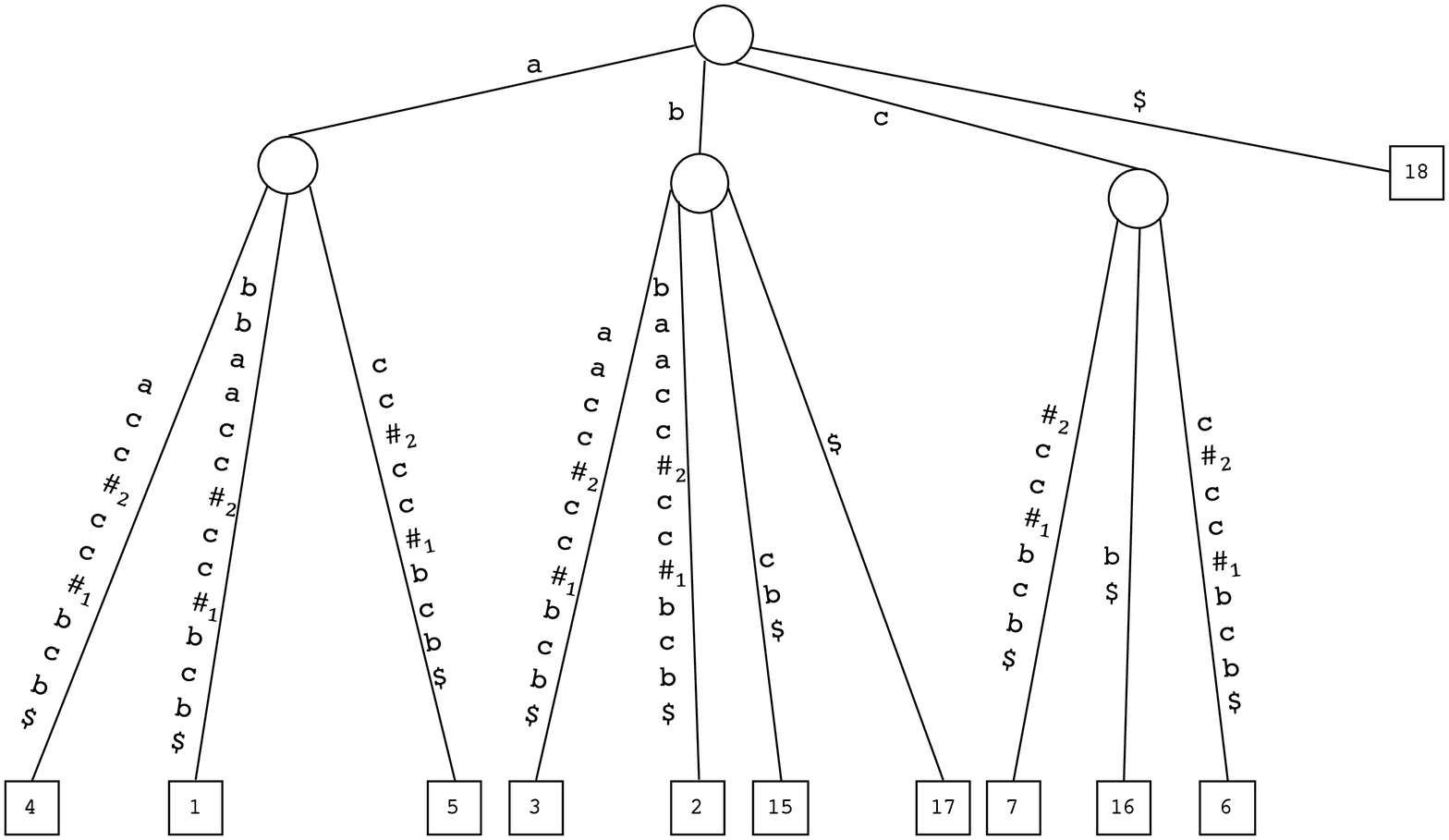}
  }
\end{figure}

Here we have a trimmed affected zone which is empty,
and hence the above tree is $\STree(w_4)$.

The current string $w_4$ contains no repeats of length at least two
consisting only of original characters.
This can also be confirmed from $\STree(w_4)$
where all the internal nodes are of string depth 1.
Hence, the algorithm terminates.

\end{document}